\documentclass{article}
\usepackage{cite}
\usepackage{amsmath,amssymb,amsfonts,amsthm}
\usepackage{algorithmic}
\usepackage{graphicx}
\usepackage[caption=false]{subfig}
\usepackage{textcomp}
\usepackage{xcolor}
\usepackage{tikz}
\usepackage{mathrsfs}
\usetikzlibrary{arrows}
	 \usetikzlibrary{calc}

\def\BibTeX{{\rm B\kern-.05em{\sc i\kern-.025em b}\kern-.08em
		T\kern-.1667em\lower.7ex\hbox{E}\kern-.125emX}}
		
\newcommand{\overbow}[1]{
	\tikz [baseline = (N.base), every node/.style={}] {
			\node [inner sep = 0pt] (N) {$#1$};
			\draw [line width = 0.4pt] plot [smooth, tension=1.3] coordinates {
				($(N.north west) + (0.1ex,0)$)
				($(N.north) + (0,0.5ex)$)
				($(N.north east) + (0,0)$)
			};
	}
}

\newtheorem{theorem}{Theorem}
\newtheorem{lemma}{Lemma}
\newtheorem{corollary}{Corollary}
\newtheorem{definition}{Definition}

\begin{document}

% \title{Formation of Conic Sections with Faulty Robots%*\\
% % {\footnotesize \textsuperscript{*}Note: Sub-titles are not captured in Xplore and
% % should not be used}
% %\thanks{Identify applicable funding agency here. If none, delete this.}
% }

\title{Conic Formation in Presence of Faulty Robots%*\\
\thanks{Full version of the conference paper in ALGOSENSORS, 2020.}
}

\author{
Debasish Pattanayak\thanks{Department of Mathematics, IIT Guwahati, Guwahati, India} \thanks{Visit to University of Vienna is supported by the Overseas Visiting Doctoral Fellowship, 2018 Award No. ODF/2018/001055 by the Science and Engineering Research Board (SERB), Government of India.}
\and
Klaus-Tycho Foerster \thanks{Faculty of Computer Science, University of Vienna, Vienna, Austria} 
\and
Partha Sarathi Mandal \thanks{Department of Mathematics, IIT Guwahati, Guwahati, India} 
\and
Stefan Schmid \thanks{Faculty of Computer Science, University of Vienna, Vienna, Austria}
}
% \institute{
	% Department of Mathematics, IIT Guwahati, Guwahati, India
% \and
	% Faculty of Computer Science, University of Vienna, Vienna, Austria
% }
\date{}
\maketitle

\begin{abstract}
	Pattern formation is one of the most fundamental problems in distributed computing, which has recently received much attention.
	In this paper, we initiate the study of distributed pattern formation in situations when some robots can be  \textit{faulty}. In particular, we consider the 
	%usual and challenging 
	well-established \textit{look-compute-move} model with oblivious, anonymous robots.
	We first present lower bounds and show that any deterministic algorithm takes at least two rounds to form simple patterns in the presence of faulty robots.
	We then present distributed algorithms for our problem  which match this bound, \textit{for conic sections}:
	in at most two rounds, robots form lines, circles and parabola tolerating $f=2,3$ and $4$ faults, respectively.
	For $f=5$, the target patterns are parabola, hyperbola and ellipse.
	We show that the resulting pattern includes the $f$ faulty robots in the pattern of $n$ robots, where $n \geq 2f+1$, and that $f < n < 2f+1$ robots cannot form such patterns.
	We conclude by discussing several relaxations and extensions. 

\end{abstract}

\section{Introduction}

Self-organizing systems have fascinated researchers for many decades already.
These systems are capable of forming an overall order from an initially disordered configuration, using \textit{local} interactions between its parts only.
Self-organization arises in many forms, including physical, chemical, and biological systems, and can be based on various processes, from crystallization, over chemical oscillation, to neural circuits~\cite{selforg}.
Due to their decentralized and self-healing properties, self-organizing systems are often very robust.

We, in this paper, consider self-organizing systems in the context of \textit{robotics}. In particular, we are interested in the fundamental question of how most simple robots can self-organize into basic patterns. 
This \textit{pattern formation} problem has already received much attention in the literature~\cite{FlocchiniPSW08,BramasT16,SuzukiY99,TomitaYKY17}.

A particularly well-studied and challenging model is 
%known as 
the \textit{look-compute-move} model, in which each robot, in each round, first observes its environment and then decides on its next move. 
In the most basic setting, the robots do not have any persistent memory and hence cannot remember information from previous rounds, and they can also not distinguish between the other robots they see: the robots are \textit{oblivious} and \textit{anonymous}.
Furthermore, robots cannot communicate.
Over the last years, several research lines investigated
%have been obtained on 
when robots can and cannot form different patterns~\cite{SuzukiY99,BramasT16,TomitaYKY17,YamauchiUKY17,FlocchiniPSW08,0001FSY15,FujinagaYOKY15,YamauchiY13}. 
However, existing work on pattern formation shares the assumption that robots are non-faulty.

This paper initiates the study of distributed pattern formation algorithms for scenarios where some robots can be \textit{faulty}: the faulty robots do not move nor act according to a specified protocol or algorithm.
The setting with faulty robots is particularly challenging, as the non-faulty robots cannot directly observe which robots are faulty.
However, even \textit{indirect} observations seem challenging: since all robots are oblivious, a robot cannot remember patterns from previous rounds, and hence, has no information which robots
moved recently. % and which not. 
In fact, a robot \textit{per se} does not even know whether the current pattern it observes is the initial configuration or whether some rounds of movements already occurred.
What's more, the ability to self-organize into a specific pattern seems to require some coordination or even consensus, which is notoriously hard in faulty~environments.

%\subsection{Our Contributions}

% \smallskip

\noindent\textbf{Contributions.}
This paper considers the design of distributed algorithms for pattern formation of most simple oblivious and potentially \textit{faulty} robots. 
Our main result is an algorithmic framework that allows robots to form patterns which include faulty robots, in a decentralized and efficient manner. 
In particular, we do not require robots to identify faulty robots explicitly or to remember previous configurations, but require knowledge of the exact number of faults.

% Depending on the number of faults, $f\in\{2,3,4,5\}$, we show how to 
% form specific target patterns such as the line, circle, parabola, hyperbola and ellipse, in just \emph{two rounds}, for at least $2f+1$ robots.
For $f$ faults, we show how to 
form conic patterns in just \emph{two rounds}, for at least $2f+1$ robots.
We form conic patterns such as line, circle and parabola for $f=2,3$ and $4$, respectively. For $f=5$, the target pattern are parabola, hyperbola and ellipse.
We also prove that this is optimal: no deterministic algorithm can solve this problem in just one round or with less than $2f+1$ robots.
We further discuss several relaxations of our model and extensions of our results, 
e.g., considering initial symmetric configurations, having at most $f$ faulty robots, or the impossibility of forming the pattern corresponding to $f$~faults. 
We also discuss an extension where the robots form a line (a circle) for $f=3,4,5$ ($f=4,5$).

% \smallskip
\noindent\textbf{Organization.}
After discussing related work in \S\ref{subsec:rel},
we first provide a formal model in \S\ref{sec:prelim}, followed by 
a study of the special case of $f=1$ faulty robot in \S\ref{sec:example} to provide some intuition.
We then give tight runtime and cardinality lower bounds in \S\ref{sec:lowerbound}, and match them for the remaining conic patterns in \S\ref{sec:algos}.
In \S\ref{sec:correctness}, we show the algorithmic framework and prove the correctness of our algorithm.
After discussing further model variations in \S\ref{sec:discussion}, we conclude in \S\ref{sec:conclusion}.

\subsection{Related Work}\label{subsec:rel}
% \TODO{elaborate}
% \stefan{debasish, please restructure a bit:
% first say something quite general, like `
% Then go along the lines of what you proposed, however,
% whenever possible, add one sentence about what is
% different in our paper (e.g., we consider a different model,
% a different algorithm, a different technique...)
% }

Pattern formation is an active area of research~\cite{FujinagaYOKY15,YamauchiY13,0001FSY15},
however, to the best of our knowledge, we are the first to consider pattern formation in the presence of faults: a fundamental extension.
In general, pattern formation allows for exploring the limits of what oblivious robots can compute.
A ``weak robot'' model was introduced by Flocchini et al.~\cite{FlocchiniPSW08}, where the objective is to determine the minimum capabilities a set of robots need to achieve certain tasks. In general, the tasks include \textit{Gathering}~\cite{ChaudhuriM15,2012Flocchini,CieliebakFPS12}, \textit{Convergence}~\cite{Cohen2005,CohenP04}, \textit{Pattern Formation}~\cite{SuzukiY99,FujinagaYOKY15,BramasT16}, etc.
Gathering is a special case of pattern formation, where the target pattern is a point.
Gathering has been achieved for robots with multiplicity detection~\cite{FlocchiniPSW05,CieliebakFPS12}.
Most gathering algorithms use the capability of multiplicity detection to form a unique multiplicity point starting from a scattered configuration. In the absence of this capability, it has been proved that gathering is impossible in the semi-synchronous model without any agreement on the coordinate system~\cite{Prencipe2007}.
% \klaus{I guess their model is different? In our case, if we can elect a leader, we can do it?}

The objective of gathering algorithms is only to gather the non-faulty robots, not to form general patterns. Moreover, for the specific case of gathering, some interesting first fault-tolerance studies exist. 
Agmon and Peleg~\cite{AgmonP06} solve the gathering problem for a single crash-fault.
Gathering has been solved with multiple faults~\cite{Bouzid0T13,BramasT15} with strong multiplicity detection. Next, gathering has also been addressed for robots with weak multiplicity detection tolerating multiple crash faults~\cite{BhagatM17,PattanayakMRM19}.
For byzantine faults, Auger et al.~\cite{AugerBCTU13} show impossibility results, and Defago et al.~\cite{DefagoP0MPP16} present a self-stabilizing algorithm for gathering.
The gathering algorithms only gather non-faulty robots.
Since oblivious robots cannot differentiate a faulty robot from a non-faulty one, all the algorithms can be considered to be non-terminating algorithms. In contrast in
our paper, we include the faulty robots in the pattern, and as a result, we achieve termination.

Flocchini et al.~\cite{FlocchiniPSW08} characterize the role of common knowledge, like agreement on the coordinate system as a requisite for the pattern formation problem, and 
Yamashita and Suzuki~\cite{YamashitaS10} characterize the patterns formable by oblivious robots.
Fujinaga et al.~\cite{FujinagaYOKY15} present an algorithm using bipartite matching for asynchronous robots.
Yamauchi and Yamashita~\cite{YamauchiY13} propose a pattern formation algorithm for robots with limited visibility.
Das et al.~\cite{0001FSY15} characterize the sequence of patterns that are formable, starting from an arbitrary configuration. 
Das et al.~\cite{DBLP:journals/access/DasFPS20} further extend the sequence of pattern formation problem for luminous robots (robots with visible external persistent memory).
As a special pattern formation problem, uniform circle formation has also been considered in the literature~\cite{FlocchiniPSV17}.
Formation of a plane starting from a three-dimensional configuration has also been solved for synchronous robots~\cite{YamauchiUKY17,TomitaYKY17}. The authors characterize configurations for which plane formation is possible.
The existing pattern formation algorithms consider the sequential movement of robots.
Since we consider faulty robots in our paper, all the existing algorithms are not adaptable to our cause. 
A fault-tolerant algorithm has to consider the simultaneous movement of robots and should satisfy the wait-free property to avoid cyclic dependency.

%\subsection{Organization}
% \stefan{debasish: this organization section is too detailed, please make it shorter (check any of my or Klaus papers for an example):however, you can maybe move some of this text to the contribution section?}

%\smallskip
%\noindent\textbf{Organization.}
%%
%The remainder of our paper is structured as follows.
%%
%We first provide a formal model in \S\ref{sec:prelim}, followed by 
%a study of the special case of $f=1$ faulty robot in \S\ref{sec:example} to provide some intuition.
%%
%We then give tight runtime and cardinality lower bounds in \S\ref{sec:lowerbound}, and match them for the remaining conic patterns in \S\ref{sec:algos}.
%%
%In \S\ref{sec:correctness}, we show the algorithmic framework and prove the correctness of our algorithm.
%After discussing further model variations in \S\ref{sec:discussion}, we conclude in \S\ref{sec:conclusion}.

\section{Preliminaries}\label{sec:prelim}
	
	We follow standard model assumptions, inspired by existing work, e.g.,~\cite{AgmonP06,BhagatM17,PattanayakMRM19}.
	\subsection{Model}\label{subsec:model}
    Each robot is a dimensionless point robot.
    The robots are homogeneous: they execute the same deterministic algorithm, are indistinguishable and anonymous (no identifiers), oblivious (no persistent memories), and silent (no communication). 
    The robots do not share a common coordinate system and observe others in their own \textit{local coordinate system}.
    The robots have unlimited visibility, and the determined locations of other robots are precise.
    
    Each robot follows the \textit{look-compute-move} cycle.
    A robot obtains a snapshot of relative positions of other robots with respect to its position in the \textit{look} state. Based on the snapshot of other robot positions, it decides a destination in the \textit{compute} state. In the \textit{move} state, it moves to the destination and reaches the destination in the same round. This is known as \textit{rigid} 
    robot movement.
    The scheduler, which activates the robots, follows a fully-synchronous (\textit{FSYNC}) model, i.e., all the robots look at the same time and finish movement in the same round, i.e., each completion of look-compute-move cycle is one round.
    We consider that the robots are susceptible to crash-faults, i.e., they stop moving after the crash and never recover. Moreover, the number of $f$ faulty robots is known beforehand, and as such, the robots know which types of pattern to form. In particular, we assume that the following four initial conditions hold:
    \begin{enumerate}
        \item All $f$ faulty robots have already crashed initially.
        \item All initial configurations are asymmetric.\footnote{This assumption allows us to have a unique ordering of the robots~\cite{ChaudhuriM15}.} 
        \item All robots occupy distinct positions initially.\label{ass:distinctinitial}\footnote{As any set of non-faulty robots that share a position will always perform the same actions from then on and be indistinguishable from each other.} 
        \item The faulty robots form a convex polygon.
    \end{enumerate}
    The last assumption needs the faulty robots to be at corners of a convex polygon; the non-faulty robots can lie at any position. The rationale behind the assumption is that four robots forming a triangle with a robot inside the triangle do not correspond to any conic section in $\mathbb{R}^2$, similarly, for five robots. For three or more robots, a collinear configuration is addressed in \S\ref{sec:discussion}. For two robots, the assumption trivially holds.
% \vspace{-1em}
	\subsection{Notations}
	
		A configuration $\mathcal{C} = \{p_1, p_2, \cdots, p_n\}$ is a set of $n$ points on the plane $\mathbb{R}^2$, where $p_i = (x_i, y_i)$ is a tuple representing the $x$- and $y$-coordinates of the robots. Since each robot is initially located at distinct points, it then holds that $p_i \neq p_j$ for any pair of $i$ and $j$ such that $i \neq j$.
		$f$ is the number of faulty robots.
		We will always uphold this condition in our algorithms, except for the case of $f=1$, where the target pattern is a point. 
		The target patterns are conic sections that satisfy the second degree general equation 
		$$a_1x^2 + a_2y^2 + a_3xy + a_4x + a_5y + a_6 = 0~.$$
		Depending on the values of $a_i$ for $i\in \{1,2,\ldots,6\}$, the equation represents line, circle, parabola, hyperbola or ellipse. We say a set of points form a conic pattern when they lie on the same conic section. 
		Now, we say the conic pattern passes through the set of points.
		We denote $\mathcal{P}$ as the length of the pattern and $u$ as the uniform distance.
	\section{Problem Statement and Intuition}\label{sec:example}
%		\subsection{Objective}
		\noindent\textbf{Objective.} Given a set of robots on the plane as defined in the model~(\S\ref{subsec:model}), we want the robots to form a conic pattern corresponding to the number of faults.\footnote{That is a point for $f=1$, a line for $f=2$, a circle for $f=3$, a parabola for $f=4$, and an ellipse or parabola or hyperbola for $f=5$.}
		
		%\subsection{Point Formation ($f=1$)}\label{subsec:f1}
		\smallskip
		
		\noindent\textbf{Point Formation ($f=1$).} %In order 
		To provide some intuition, we start with the case of $f=1$.
		%the simpler case of $f=1$.
		For a single faulty robot, we move all the robots to the center of their smallest enclosing circle. If there is a faulty robot in the center, then point formation is achieved. If the faulty robot is somewhere else, we arrive at a configuration with two robot positions.
		For a configuration with two robot positions, all robots move to the other robot's position. 
		From Assumption~\ref{ass:distinctinitial}, we have all the robots at distinct initial positions. Hence the faulty robot is at a different position from the gathered robots. Moving all gathered robots to the faulty robot's position achieves our objective of point formation, as the faulty robot cannot~move.
		
	\section{Lower Bound}\label{sec:lowerbound}
		
	%From our initial example with $f=1$, 
	We saw above that there could be situations where only one round suffices, namely, for the case of exactly $n=2$ and $f=1$.
	However, we can show that for $f\geq 2$, at least two rounds are required.
	For conic patterns (with $f\in \{2,3,4,5\}$), this bound is tight:
	we will later provide algorithms that terminate in two rounds.

	\begin{theorem}
			For every $f\geq 2$ and every $n\geq f+3$ holds: Any deterministic algorithm needs more than one round to make a pattern passing through all $f$ faulty robots.%for $f\geq 2$.
	\end{theorem}
	\begin{proof}
		Let $\varphi$ be a deterministic algorithm that forms the pattern with faulty robots. Suppose $\varphi$ solves the pattern in one step. 
		Two patterns can have at most $f$ common points\footnote{Two parabolas intersect at 4 points, which can be the common points between two parabola patterns.}. 
		Let $\mathcal{C} = \{p_1, p_2,\cdots, p_{f+3}\}$ be an initial configuration with $f+3$ robots such that no $f+1$ robots are in the same pattern.

		Without loss of generality, consider two sets of $f$ faulty robots at positions $\{p_1, p_2, \cdots , p_f\}$ and $\{p_2, p_3, \cdots , p_{f+1}\}$ and the corresponding pattern be $\wp'$ and $\wp''$, respectively. The $f$ faulty robots do not move. 
		Let $\wp' = \{p_1',p_2',\cdots,p_{f+3}'\}$ and $\wp'' = \{p_1'',p_2'',\cdots,p_{f+3}''\}$. As $\varphi$ achieves pattern formation in one round, both $\wp'$ and $\wp''$ should be final.
		We have $p_{f+1} \neq p_{f+1}'$, $p_{f+2} \neq p_{f+2}'$ and $p_{f+3} \neq p_{f+3}'$, since all robots in $\wp'$ are in the pattern. 
		Similarly, we also have $p_1 \neq p_1''$,  $p_{f+2} \neq p_{f+2}''$ and $p_{f+3} \neq p_{f+3}''$ for $\wp''$.
		Since the robots at $\{p_2,\cdots,p_f\}$ did not move to form $\wp'$ and $\wp''$, these $f-1$ points are common between $\wp'$ and $\wp''$. 
		Since, $p_{f+1} \neq p_{f+1}'$ and $p_{f+1} = p_{f+1}''$, so $p_{f+1}$ cannot be a common point between $\wp'$ and $\wp''$. Out of $p_{f+2}$ and $p_{f+3}$, at most one can be a common point in the pattern, since there are at most $f$ common points. 
		Since $\varphi$ is deterministic, the destination for robots at $p_{f+2}$ and $p_{f+3}$ remains the same regardless of the destination pattern being $\wp'$ or $\wp''$. If $\{p_{f+2}',p_{f+3}'\} = \{p_{f+2}'',p_{f+3}''\}$, then $\wp'$ and $\wp''$ have $f+1$ common points. This is a contradiction since the patterns are different. Hence no deterministic algorithm can solve the pattern formation problem with faults in one round.
		The arguments hold analogously for $n\geq f+3$~robots.
	\end{proof}
	
	Next, we show a lower bound on the number of robots required to solve the pattern formation problem. A configuration with exactly $f$ robots is trivially solvable since the $f$ robots are already in the pattern. 
	Note that for $f\geq 2$, $2f+1 \geq f+3$ holds. 
	% The idea behind the proof is that we need at least $f+1$ non-faulty robots to get a majority and determine the current pattern of a configuration, which in turn lets us determine the pattern passing through the faulty robots.
	% The full proof is deferred to the Appendix.
	 
	\begin{theorem}\label{thm:2f1}
		At least $2f+1$ robots are required to form a pattern passing through $f$ faulty robots for $f\geq 2$.
	\end{theorem}   
	\begin{proof}
		Consider the number of robots to be $f < n < 2f+1$. Assume that the configuration of these $n$ robots is such that no $f+1$ robots are in the same pattern. Let $\varphi$ be a deterministic algorithm, which decides the destination of the robots given a configuration. Since the robots are oblivious, it is impossible to determine which robots are faulty given a configuration. As we consider patterns from the conic section, a pattern corresponding to $f$ can be uniquely determined through $f$ robots.
		%only if there are at least $f$ robots.

		Let $\varphi$ decide the target pattern corresponding to a set of $f$ robots for the given configuration $\mathcal{C}$. So the other $n-f$ robots have to move to the pattern. Since the algorithm cannot determine which robots are faulty, the adversary can always choose the $0<n-f$ robots to be faulty. Since $n-f \leq f$, none of the robots move. This leads to a stagnated configuration, and the algorithm does not proceed further.

		If the algorithm decides a pattern that passes through less than or equal to $f$ points in the configuration, then we choose faulty robots out of the points which are not on the pattern, and we arrive at a configuration where not all the robots are in the same pattern. 
		Now there are at most $n-f \leq f$ robots on the pattern, which is the same as the previous configuration. %Hence the proof.
	\end{proof}
	% \vspace{-2em}		

		\section{Detailed Algorithms}\label{sec:algos}
	% \vspace{-1em}		
	In this section, we provide the promised two round algorithms for the different 
		%types of 
		configurations.
		To this end, we first provide algorithmic preliminaries ~in~\S\ref{subsec:algo-prem}.
% \vspace{-1em}		
		\subsection{Algorithmic Preliminaries}
		\label{subsec:algo-prem}
		In our algorithms, we will perform case distinctions according to the following three types of configurations:
		\begin{definition}[\textit{Terminal Configuration}]
			A configuration is a terminal configuration if all the robots form the target pattern corresponding to $f$ faulty robots.
		\end{definition}
		\begin{definition}[\textit{Type I Configuration}]
			If exactly $n-f$ robots are in the target pattern corresponding to $f$ faulty robots, then it is a Type I configuration.
		\end{definition}
		A Type I configuration can be symmetric or asymmetric.
		\begin{definition}[\textit{Type O Configuration}]
			If a configuration is not Terminal or Type I, then it is a Type O configuration.
		\end{definition}
		Note that an initial asymmetric configuration can be a Type I or Type O or Terminal configuration. In  the following, we also distinguish the configurations based on the uniform spacing between the robots. We use the uniform positions of the robots as a differentiating factor between faulty and non-faulty.
		\begin{definition}[\textit{Uniform Configuration}]
			A configuration is a uniform configuration if the distance between all consecutive pairs of robots in the configuration along the pattern is the same.
		\end{definition}
		\begin{definition}[\textit{Quasi-Uniform Configuration}]
			If a uniform configuration with $m$ uniform positions is occupied by $n$ robots, where $n\leq m\leq 2n$, then it is a quasi-uniform configuration.
		\end{definition}
		With the assumption (\S\ref{subsec:model}) that the initial configuration is asymmetric, we can obtain an ordering of the robots using the 
		%ordering 
		algorithm by Chaudhuri et al.~\cite{ChaudhuriM15}.
		\begin{lemma}\label{lem:orderable}
			An asymmetric configuration is orderable.~\cite{ChaudhuriM15}
		\end{lemma}
		Using Lemma \ref{lem:orderable}, we can thus always obtain an ordering among the robots. We use this ordering to determine the target pattern in case of a Type O configuration. 
		In general, having an ordering allows us to have complete agreement on the coordinate system, i.e., the robots agree on the direction and orientation.

		Let $\mathcal{O}$ be an ordering of the robots that maps the set of robots to a set of integers $\{1,2,\ldots, n\}$ such that each robot corresponds to an integer. This is the rank of the robot in the ordering. 
		In case of symmetry, two robots can have different orderings. 
		%We 
		Note that locally, the ordering is unique for a particular robot, but that from a global perspective, different robots can have different~orderings.

		We use the ordering to determine a target pattern only in cases where there are multiple potential target patterns. We always choose the target pattern passing through the smaller ranked robot in the ordering.

		Since the algorithm takes at most two steps to reach a pattern containing all faulty robots, we denote the \textit{initial}, \textit{transitional} and \textit{final} configurations as $\mathcal{C}_0$, $\mathcal{C}_1$ and $\mathcal{C}_2$, respectively.

		% The algorithm uses four parameters to determine the uniform points on the target pattern. The first is the base point. The base point of the target patterns are following.
		% \begin{itemize}
		% 	\item \textit{Line:} The midpoint of the target line.
		% 	\item \textit{Circle:} A point on the perimeter where the line joining the center of the target circle and the robot with smallest order intersect.
		% 	\item \textit{Ellipse:} The end-point of semi-minor axis on the side of the robot with smallest order.
		% 	\item \textit{Parabola or Hyperbola:} The vertex of parabola or hyperbola.
		% \end{itemize}
		% The second is the length of the pattern $\mathcal{P}$.
		% The third is the uniform distance $u$. The uniform distance is determined in such a way that there are no uniform points lies on the intersection of the pattern in the existing configuration and the target pattern.
		% % We choose the uniform distance $u = \mathcal{P}/n$.
		% The fourth is the \textit{offset}. The offset is the smallest distance of a uniform point along the perimeter of the target pattern from the base point.
		% \vspace{-1em}
		
		\subsection{Algorithm}
		We first present a general algorithm with two different strategies for the open pattern and closed pattern. Among the conic patterns, line, parabola and hyperbola are open patterns, while circle and ellipse are closed patterns. The target position of the robots depend on this since open conics have two end points, while closed do not have any. 
		The length of the pattern, $\mathcal{P}$ is determined with respect to the pattern being formed. For line, $\mathcal{P}$ is the distance between the two endpoints. For parabola (hyperbola), $\mathcal{P}$ is the length of the parabola (hyperbola) between the points where the latus rectum\footnote{The latus rectum is the line that passes through the focus of the parabola and parallel to the directrix. 
		% A common form of representing a parabola is $x^2 = 2ly$, where $2l$ is the length of the latus rectum.
		} of parabola (hyperbola) intersects the parabola (hyperbola). In case of circle and ellipse, $\mathcal{P}$ is the perimeter.
		We denote $u$ as the uniform distance at which the points on the target pattern are determined. The length $u$ is computed along the pattern. 
		We differentiate between two types of configurations:
		
		% \vspace{-1em}
		\subsubsection{Type O Configuration:} In this case, we form the target pattern outside the smallest enclosing circle of the configuration. The target pattern can be uniquely determined using the asymmetricity of the configuration, since a Type O configuration can only appear in the initial state. The target pattern size is dependent on the diameter of smallest enclosing circle. 
		\begin{itemize}
			\item Compute the smallest enclosing circle of the configuration where $O$ is the center and $d$ is the diameter.
			\item Let $A$ be the location of the robot with the smallest rank in the ordering. If $A$ is $O$, then we choose the robot with second smallest rank.
			\item Find point $B$ such that $B$ lies on $\overrightarrow{OA}$ and $|\overline{OB}| = d$.
			\item The target pattern corresponding to $f$ passes through $B$.
		\end{itemize}
		\begin{figure}\centering
			% \vspace{-2em}
			\subfloat[$f=2$
				\label{fig:typeOf2}]{\centering
				\includegraphics[height=0.2\linewidth]{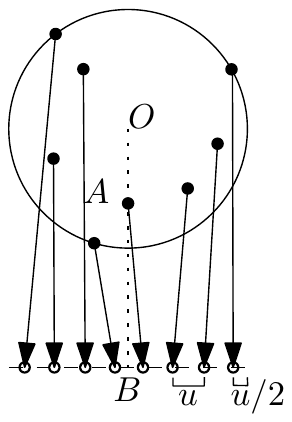}
				}
			\subfloat[$f=3$\label{fig:typeOf3}]{
				\includegraphics[height=0.2\linewidth]{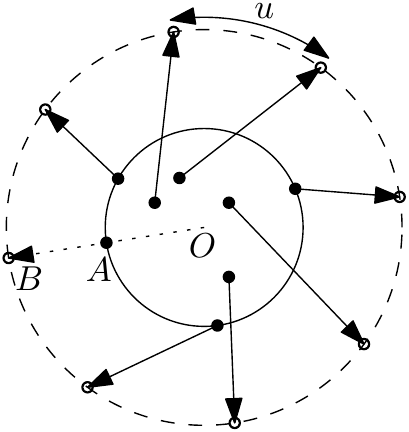}}
			\subfloat[$f=4$ and $f=5$\label{fig:typeOf4}]{
				\includegraphics[height=0.2\linewidth]{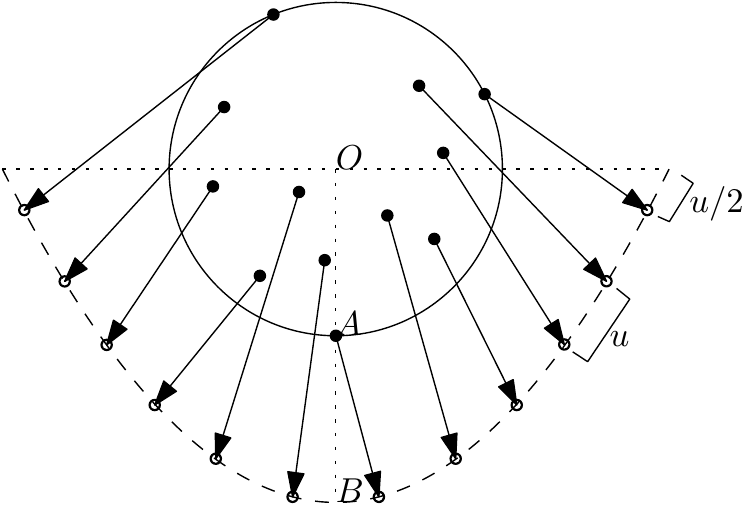}
			}
			% \subfloat[$f=5$\label{fig:typeOf5}]{\includegraphics[clip, trim =22cm 10cm 21cm 10cm, height=0.15\linewidth]{circleEllipse.pdf}}
			\caption{Target patterns (dashed) for Type O configuration}
			\label{fig:typeO}
			% \vspace{-2em}
		\end{figure}
		Now, we show how we determine the target pattern corresponding to the number of faults (see Fig.~\ref{fig:typeO}).
		\begin{description}
			\item[$f=2$:] The target line is perpendicular to $\overline{OB}$ and has its midpoint at $B$ with length $d$ (ref. Fig.~\ref{fig:typeOf2}).
			\item[$f=3$:] The target circle passing through $B$ has radius $d$ and center at $O$ (ref. Fig.~\ref{fig:typeOf3}).
			\item[$f=4$ and $f=5$:] The target parabola has its vertex at $B$ and focus at $O$. The latus rectum of parabola is perpendicular to $\overline{OB}$ and length of latus rectum is $2d$ (ref. Fig.~\ref{fig:typeOf4}).
			% \item[$f=5$:] The center of ellipse is at $O$ and one endpoint of semi-minor axis is at $B$ (ref. Fig.~\ref{fig:typeOf5}).
		\end{description}
		For an open pattern, 
		we choose the first point at a distance $u/2$ from one end point and place subsequent $n-1$ points at distance $u = \mathcal{P}/n$ (ref. Fig.~\ref{fig:typeOf2}, \ref{fig:typeOf4}).
		% choosing $u = \mathcal{P}/n$ results in $n+1$ points, while $n$ points for closed patterns.
		% We shift the points $u/2$ inwards for line and parabola to arrive at $n$ pattern points .
		
	\subsubsection{Type I Configuration:} The target pattern is determined as passing through the robots which are not part of the pattern in the existing configuration. The destinations for robots are at points uniformly positioned at distance $u$ corresponding to the length of the pattern. However, there is a possibility that the existing pattern and target pattern intersect. We hence avoid using the intersection points as target points. If an intersection point is a target point for a value of $u$, then we choose $u'$ depending on the configuration.
    \begin{figure}[h!t]\centering
        % \vspace{-3em}
		\subfloat[$f=2$\label{fig:asymTypeIf2}]{
			\includegraphics[width=0.15\linewidth]{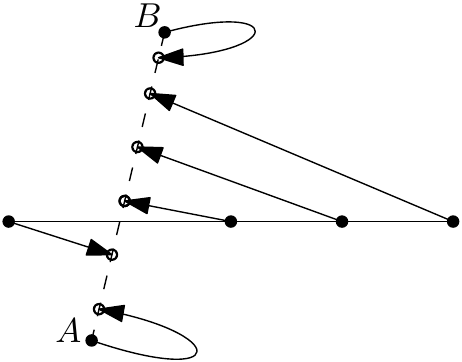}
			}
		\subfloat[$f=3$\label{fig:asymTypeIf3}]{
			\includegraphics[width=0.2\linewidth]{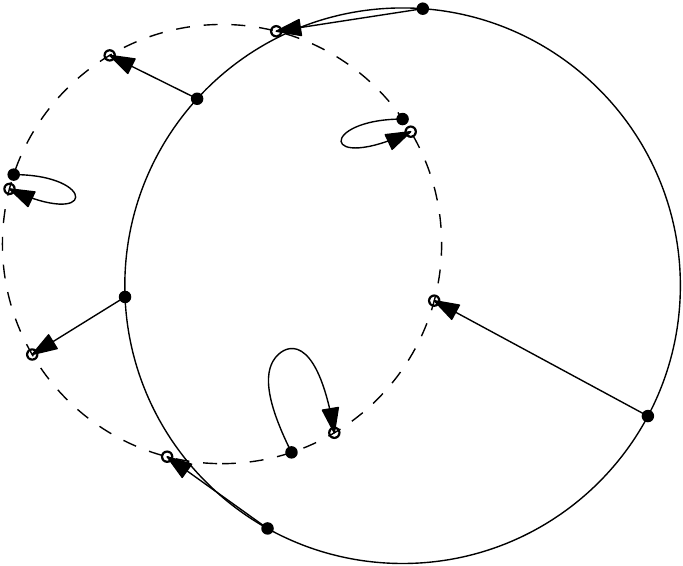}
		}
		\subfloat[$f=4$\label{fig:asymTypeIf4}]{
		\includegraphics[width=0.2\linewidth]{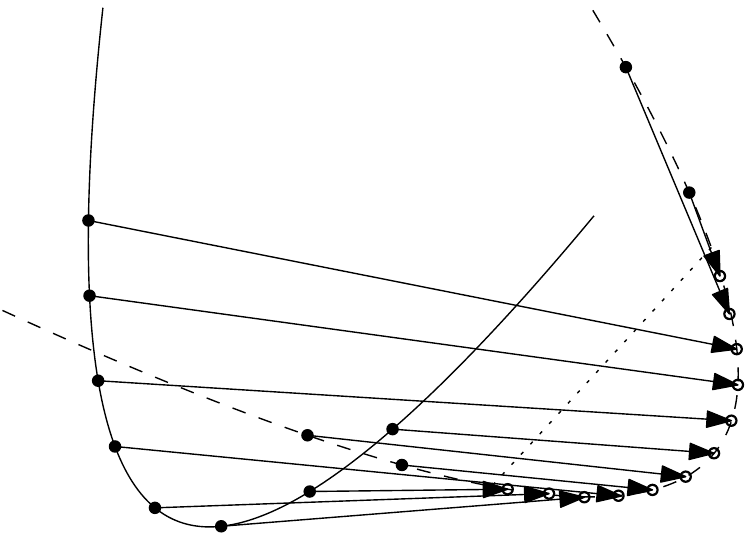}
		}
		\subfloat[$f=5$\label{fig:asymTypeIf5}]{
		\includegraphics[width=0.2\linewidth]{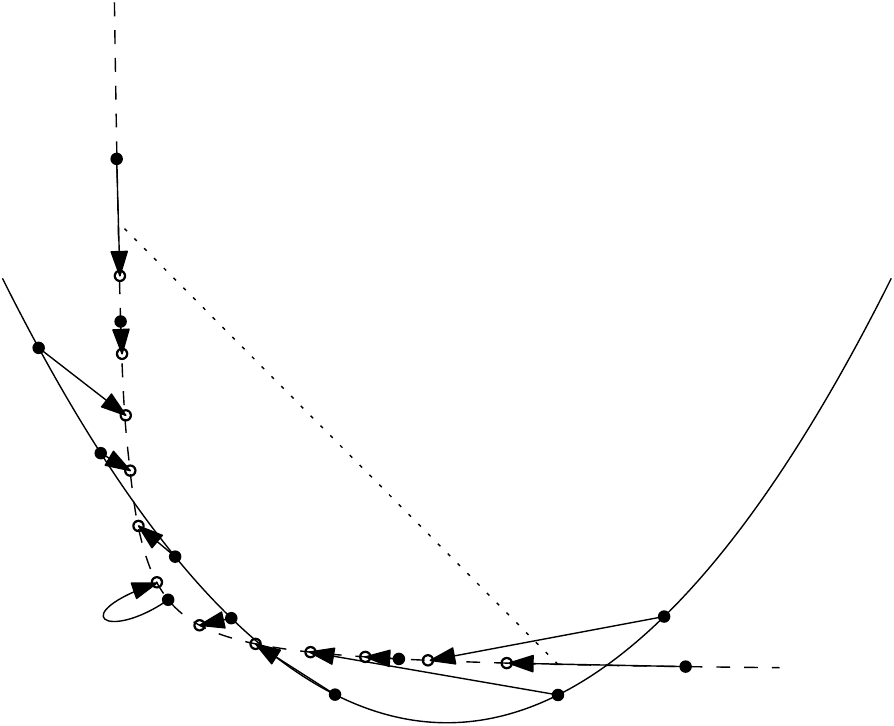}
		}
		\caption{Target patterns (dashed) for Asymmetric Type I configurations}\label{fig:asymTypeI}
		% \vspace{-3em}
	\end{figure}
		\begin{description}
			\item[Asymmetric Type I:] The robots are assigned target points according to their rank in the ordering with $u = \mathcal{P}/n$ (ref.~Fig.~\ref{fig:asymTypeI}).
			If the target point corresponding to $u$ overlaps with intersection points, then we choose $u' = \mathcal{P}/(n+1)$. The number of target points with respect to $u'$ is $n+1$, so we assign two destinations to the robot with highest rank, which can be chosen arbitrarily~(ref.~Fig.~\ref{fig:non-overlapping}).
			For $f=5$, the intermediate configuration can be a parabola, hyperbola or an ellipse. Thus, it can also appear as an initial configuration.
			We show all the transition between ellipse, parabola and hyperbola in the Appendix.
		
			\item[Reflective Symmetric Type I:] Let $k$ be the number of robots which lie on the line of symmetry. In this case, we choose, $u = \mathcal{P}/(n+k)$. We get $k$ extra target points so that we can assign two target points to the robots on the line of symmetry. A robot on one side (say left) of the line of symmetry finds its destination on the same side. Since the robots on the line of symmetry may not have a common left or right, they can choose one of the two symmetric destinations as their target. This also ensures that the next configuration is not completely symmetric (ref. Fig.~\ref{fig:refsymTypeI}), since only one of the symmetric points would be occupied. On overlap of target pattern points with existing points, we choose $u'$ similar to asymmetric Type I configuration. We show other transitions for $f=5$ in the Appendix.
		\begin{figure}
			% \vspace{-3em}
			\centering
			\subfloat[$f=2$\label{fig:refsymTypeIf2}]{
				\includegraphics[width=0.25\linewidth]{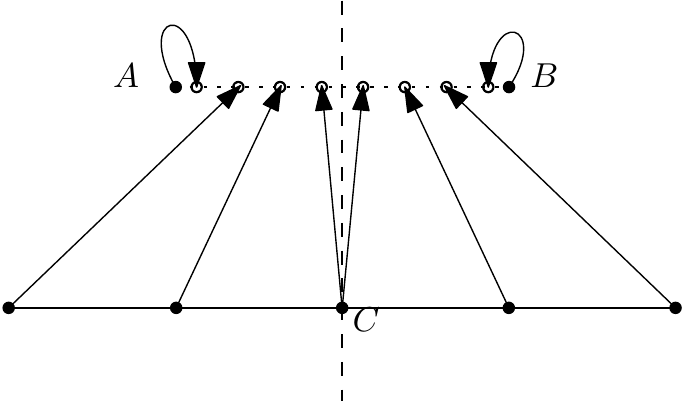}}
			\subfloat[$f=3$\label{fig:refsymTypeIf3}]{
			\includegraphics[height=0.18\linewidth]{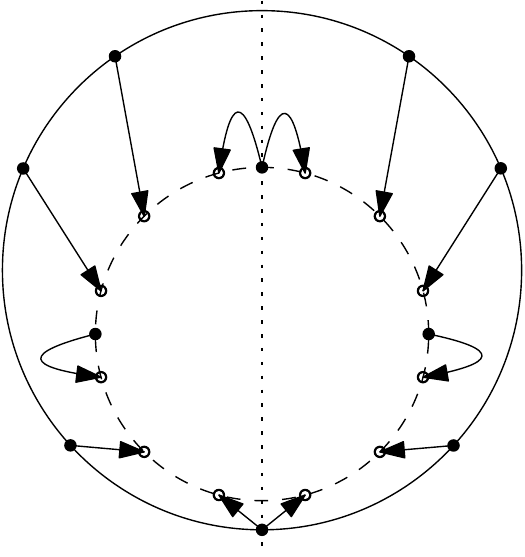}
			}
			\subfloat[$f=4$\label{fig:refsymTypeIf4}]{
			\includegraphics[width=0.25\linewidth]{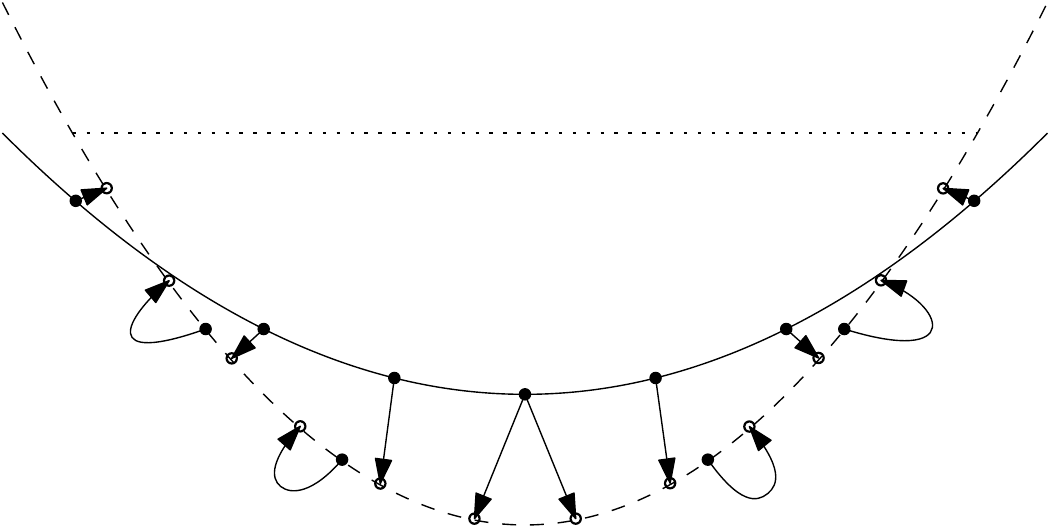}
			}
			\subfloat[$f=5$\label{fig:refsymTypeIf5}]{
			\includegraphics[width=0.25\linewidth]{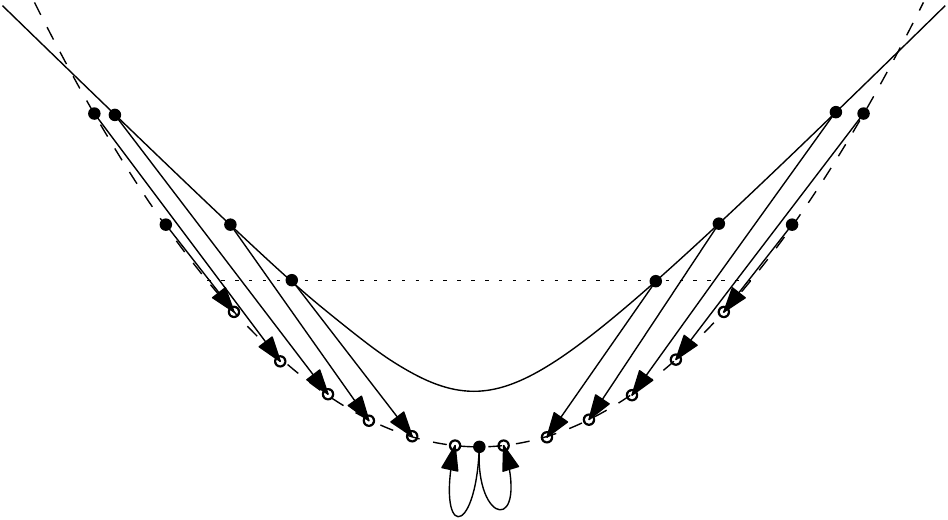}
			}
			\caption{Target patterns (dashed) for Reflective Symmetric Type I configurations}\label{fig:refsymTypeI}
			% \vspace{-2em}
		\end{figure}

			\item[Rotational Symmetric Type I:] This can occur only in the case of lines and circles. In case of a line, the robots move to the side with smaller angle between the target line $m$ and existing pattern $l$ (ref. Fig.~\ref{fig:rotationSymTypeIf2}). In case of a circle, the configuration is two concentric circles. We assign two destinations for each robot already on the target circle by choosing $u = \mathcal{P}/(n+k)$, where $k = 3$~(ref.~Fig.~\ref{fig:rotationSymTypeIf3}).
		\end{description}
		\begin{figure}[h!t]
			\centering
			% \vspace{-3em}
			\subfloat[$f=2$\label{fig:rotationSymTypeIf2}]{
				\includegraphics[height=0.15\linewidth]{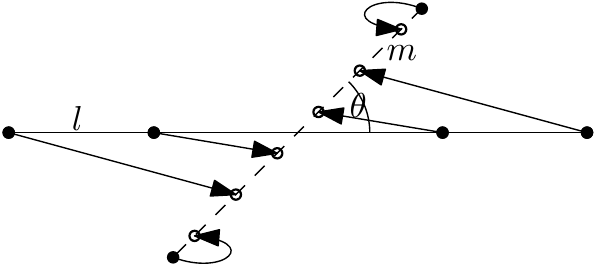}
			}
			\subfloat[$f=3$\label{fig:rotationSymTypeIf3}]{
			\includegraphics[height=0.15\linewidth]{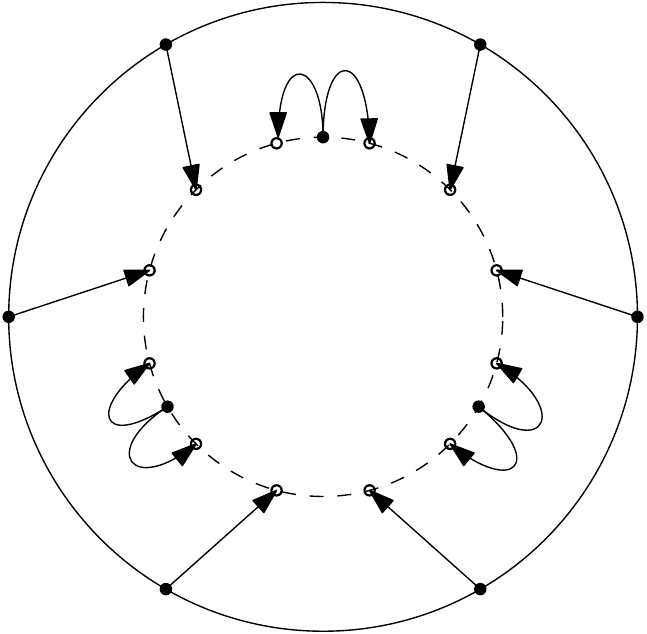}
			}
			\caption{Target patterns (dashed) for Rotational Symmetric Type I configurations}
			% \vspace{-3em}
			\label{fig:rotationSymTypeI}
		\end{figure}

		\begin{figure}[h!t]
			% \vspace{-1em}
			\centering
			% \subfloat[$f=2$\label{fig:AsymTypeIf2}]{
				\includegraphics[width=0.6\linewidth]{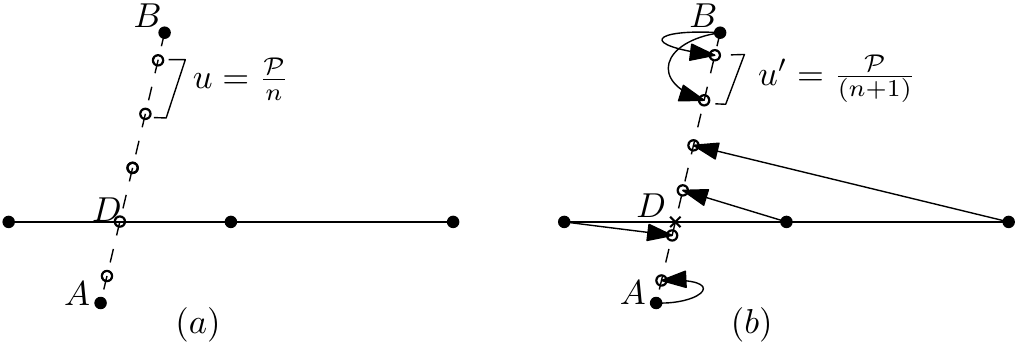}
			% }
			% \subfloat[$f=3$\label{fig:AsymTypeIf3}]{
			% 	\includegraphics[width=0.4\linewidth]{asymmetricIntermediateCircle.pdf}}
				\caption{Choosing Non-overlapping uniform points in the target pattern (dashed)}
			% \vspace{-3em}
			\label{fig:non-overlapping}
			\end{figure}
	\subsubsection{Choosing Non-overlapping Uniform Points:}
	% \subsection{Choosing Non-overlapping Uniform Points:}
	In case of an asymmetric Type I configuration, we obtain an ordering $\mathcal{O}$. 
	For $f=2$, the target line passes through $A$ and $B$.
	Let $A$ be the robot with smaller rank among $A$ and $B$. We choose the set of uniform points at a distance $u/2$ from $A$ towards $B$. Let $D$ be the intersection point of the existing line in the current configuration and the line through $A$ and $B$ ($D$ is marked as a cross in Fig.~\ref{fig:non-overlapping}b). If the uniform points overlap with $D$ (ref. Fig.~\ref{fig:non-overlapping}a), then we choose another set of points at uniform distance $u' = \mathcal{P}/(n+1)$. Since, $|\overline{AD}| = 3u/2$, the uniform points overlap with $D$.
	We have to choose a set of uniform points corresponding to $u'$ as shown in Fig.~\ref{fig:non-overlapping}b.
    With $u'$, there are $n+1$ uniform points, the robot $B$ has two potential destinations, one of which can be chosen arbitrarily.
    
    For $f=3$, the target pattern passes through the three points, which are not in the circle in the current configuration. Two circles can have at most two intersection points. Let the robots at $A$, $B$, and $C$ be the three robots that were not on the circle of the current configuration. Let $d$ be the diameter of the circle passing through $A$, $B$, and $C$. The pattern length is $\mathcal{P} = \pi d$. We choose uniform distance $u = \mathcal{P} /n$. Without loss of generality, let the order among $A$, $B$ and $C$ follow $A < B < C$. Let the intersection points of the two circles be $D$ and $E$ (ref. Fig.~\ref{fig:AsymTypeIf3}). We determine a set $\mathcal{U}$ of uniform points which correspond to $A$, $B$ or $C$, i.e., the set of points at a uniform distance from these points on the pattern. In Fig.~\ref{fig:AsymTypeIf3}, the uniform points corresponding to $A$, $B$, and $C$ are denoted with empty squares, filled squares, and crosses, respectively.
	\begin{figure}[ht]
		\centering
		% \subfloat[$f=3$\label{fig:AsymTypeIf3}]{
				\includegraphics[width=0.3\linewidth]{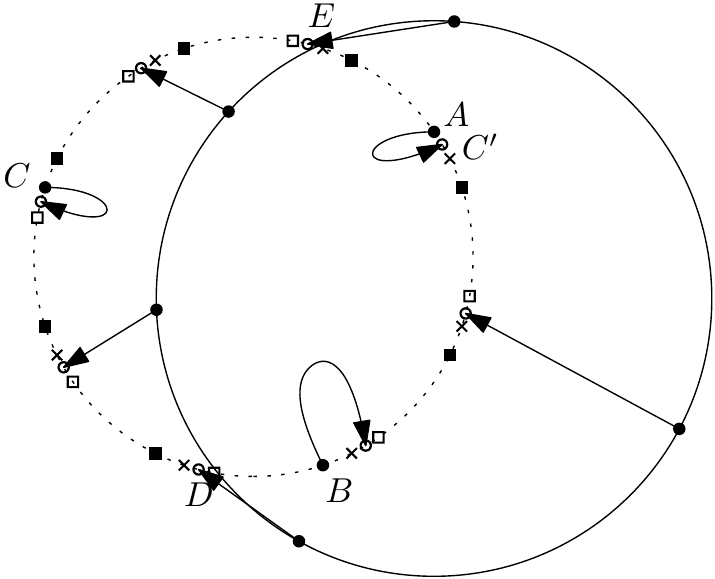}
				% }
		% \subfloat[$f=3$\label{fig:AsymTypeIf4}]{
				% }
		\caption{Choosing non-overlapping uniform points for $f=3$}
		\label{fig:AsymTypeIf3}
	\end{figure}
	Take the distance between $A$ and the closest point from $\mathcal{U}$, i.e., $C'$. Consider the target points as the uniform set corresponding to the midpoint of the arc $\overbow{AC'}$. Fig.~\ref{fig:AsymTypeIf3} shows these points marked with small circles on the target circle (dotted circle). If the uniform point set includes the intersection points $D$ or $E$, we choose the second closest point from $A$. The clockwise direction is the direction of $A$ towards $B$. The ordering of the target pattern points follows the clockwise direction starting from the point near $A$.
	If the uniform point set fails to have a non-overlapping set of points, then we expand the uniform set $\mathcal{U}$ to include points with uniform distance $u' = \mathcal{P}/(n+1)$. If there are $n+1$ destinations, then the robot with the highest rank would have two~destinations.

	\begin{figure}[ht]
		\centering
		\includegraphics[width=\linewidth]{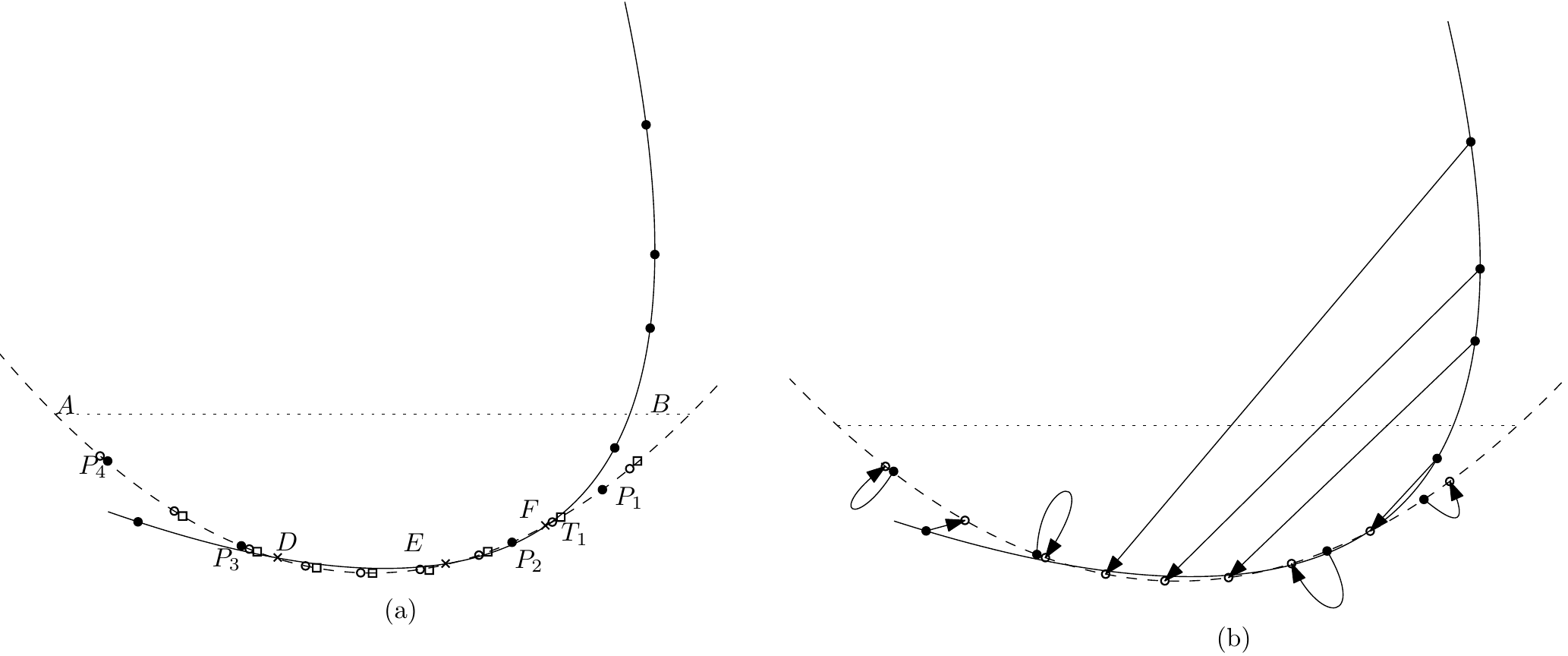}
		\caption{Choosing non-overlapping uniform points for $f=5$}
		\label{fig:AsymTypeIf4}
	\end{figure}
	For $f = 4$, we need to choose uniform points on the parabola between the points where the latus rectum intersects the parabola (Points $A$ and $B$ as shown in Fig.~\ref{fig:AsymTypeIf4}a). Unlike the case of the line, some points in the existing configuration can already lie inside the part where uniform points are to be placed. Our objective is to avoid the existing points. For Asymmetric Type I, we choose the first uniform point at a distance $u/2$ from $A$, where $u = \mathcal{P}/n$ and other uniform points consecutively at distance $u$. Let this set be the set of uniform points $\mathcal{U}$. If the uniform points overlap with any existing points in the configuration or the intersection points of the current and target pattern, then we shift the uniform points. As shown in Fig.~\ref{fig:AsymTypeIf4}, the point $P_4$ lies at a distance $u/2$ from $A$. The chosen uniform points overlap with an existing point on the configuration. Now, we can find the closest point that is at non-zero distance from the points in set $\mathcal{U}$ among the intersection points ($D$, $E$ and $F$ are marked as cross in Fig.~\ref{fig:AsymTypeIf4}a) and robot positions on the target pattern in the current configuration ($P_1$, $P_2$, $P_3$ and $P_4$ on the dashed parabola in Fig.~\ref{fig:AsymTypeIf4}a). In Fig.~\ref{fig:AsymTypeIf4}a, the point $F$ is the closest to the point  uniform set $\mathcal{U}$. Thus we shift the uniform point in $\mathcal{U}$, $T_1$ near $F$ towards $F$ to the midpoint of the curve segment $FT_1$. Likewise, we shift all the target points in $U$ by the same distance in the same direction. Thus, the resulting set of uniform points (marked as circles in Fig.~\ref{fig:AsymTypeIf4}a and \ref{fig:AsymTypeIf4}b) does not overlap with any of the existing points in the configuration or the intersection of the current and target pattern.

	For $f=5$, we use similar strategies as with $f=4$ for the case of parabola and hyperbola. In case of an ellipse, we use a strategy similar to the circle as mentioned for $f=3$.
% We defer the remaining discussion of this paragraph to the Appendix due to space constraints.
% , but provide a main idea for $f=3$ in Figure~\ref{fig:AsymTypeIf3}.

%

	\section{General Algorithmic Framework}\label{sec:correctness}

	The algorithm broadly has two steps, based on the current robot configuration.
	\begin{description}
			\item[Step~1:] Determine the faulty robots.
			\item[Step~2:] Move to a pattern passing through the faulty robots. 
	\end{description}
	Since the robots are oblivious, they cannot distinguish between Step~1 and Step~2. Hence the properties required for Step~1 have to be applicable to Step~2 and vice versa. For Step~1, a simple process to determine all the faulty robots in a single round is to move all the robots, such that the robots which do not move, will not lie on the pattern points. 
	We determine the pattern points uniformly so that all pattern points are consecutively equidistant along the pattern. This helps us in determining the faulty robots since they would not lie on a pattern point.
	For Step~2, the pattern determined from the faulty robot positions has to be unique, so that all the robots agree on the pattern. Overall, the algorithm needs to determine a unique pattern such that all robots agree on the pattern, and all robots are required to move to achieve the pattern. 
	\noindent We present a transition diagram for configurations in our algorithm in Fig.~\ref{fig:transition}.
	\begin{figure}[ht]
			% \vspace{-1em}
			\centering
		\includegraphics[width=0.6\linewidth]{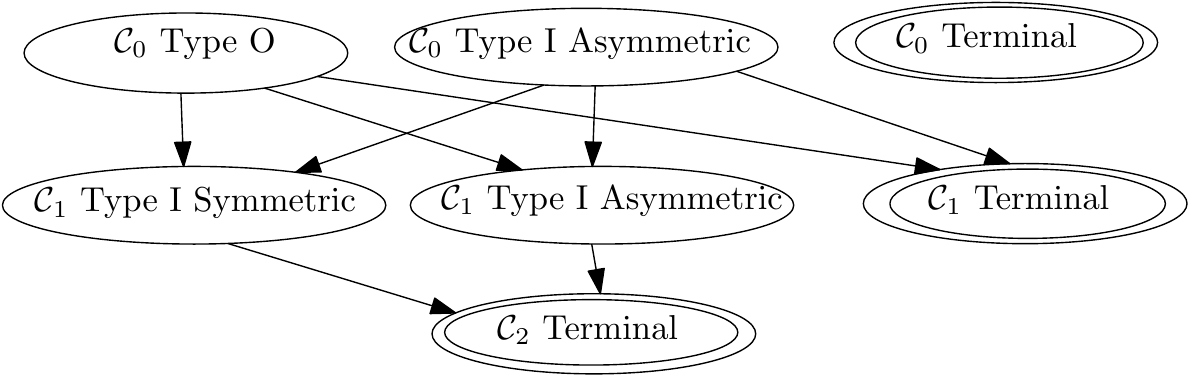}
		\caption{Transition diagram of configurations}
			% \vspace{-2em}
			\label{fig:transition}
	\end{figure}
	
	\begin{lemma}\label{lem:distinct}
		The destinations of all robots are distinct.
	\end{lemma}
	\begin{proof}
		We always follow the ordering to determine the destinations for each robot. In the asymmetric cases, we have the ordering, which creates a one-to-one map from the current position of a robot to its destination. In case of symmetric configurations, the robots which are present on the line of symmetry have two potential destinations (from a global perspective) and choose one of them according to their local orientation. Hence the destinations are distinct.
	\end{proof}
    Next, we show that there are no overlapping points between the current configuration and the set of destinations. The set of destinations is the set of all potential destination points for the robots in the current configuration. Since some robots are faulty, two consecutive configurations may have those points as common points. Had the robots been non-faulty, then they would have moved to a point that is not in the current configuration. 
    % The full proof is in the Appendix.
	\begin{lemma}\label{lem:nointersection}
			Given configuration $\mathcal{C}$ and a destination set $\mathcal{C}'$, we have $\mathcal{C}\cap \mathcal{C}' =\phi$.
	\end{lemma}
	\begin{proof}
		For a Type O configuration, the destinations are outside the smallest enclosing circle, and thus we have no intersection.
		For a Type I configuration, the uniform points are chosen such that they do not overlap with the intersection points or the potentially faulty robot positions through which the target pattern is determined. 
		Let there be $k$ points through which the target pattern is being determined for a Type I configuration. In this case, we assume that the robots are faulty, but that may not be true. So we need to make sure that the destination points do not overlap with any existing point. Otherwise, we would have two robots at a point in the resulting configuration.
		Let $u$ be the uniform distance at which the uniform positions to be chosen.
		If all $k$ points belong to the same uniform set, then we may have to choose another set of points since there is a possibility that the uniform points chosen, overlap with the intersection points of the current pattern and the target pattern. 
		So we choose a different value of $u$ and repeat the process again. For  an asymmetric Type I configuration, we choose $u' = \mathcal{P}/(n+1)$. In that case, we may obtain more points where we associate two destinations for the robot with the highest rank.
		Hence there is no intersection between the given configuration and the destination configuration.
	\end{proof}
	\subsection{Determining Faulty Robots and a Pattern}
	There are two types of initial configurations where we need to determine the faulty robots, i.e., arbitrary configurations and intermediate configurations. The destinations for the robots are such that no point in $\mathcal{C}_0$ overlaps with any point in $\mathcal{C}_1$. For an arbitrary initial configuration, the target pattern is scaled such that no point in the target pattern lies on or inside the smallest enclosing circle of $\mathcal{C}_0$. Since $\mathcal{C}_0$ is asymmetric, we can always uniquely scale the pattern. 
	We can moreover show that a unique pattern exists that passes through all the faulty robots, and 
obtain the following three lemmas.
% , with full proofs in the Appendix.

	\begin{lemma}\label{lem:initarbitfault}
		It takes one round to determine all the faulty robots for a Type O configuration for $f \in \{2,3,4,5\}$.
	\end{lemma}
	\begin{proof}
		Since we need to determine all the faulty robots, we make sure that all the robots are required to move to form the target pattern. No point in the existing configuration coincides with the points in the target pattern. According to the algorithms in the previous section, we ensure that the robots move to a point outside the smallest enclosing circle of $\mathcal{C}_0$. So in $\mathcal{C}_1$, the robots which lie inside the pattern are faulty. For $f=2$, the faulty robots would be in-between the first and the second robot of the existing line. For $f=3$, the faulty robots would lie inside the circle or ellipse. For $f=4$ and $f=5$, the robots would lie on the side of the parabola with the focus. %Hence the proof.
	\end{proof}
	\begin{lemma}\label{lem:initinterfault}
		It takes one round to determine all the faulty robots for an asymmetric Type I configuration for $f \in \{2,3,4,5\}$.
	\end{lemma}
	\begin{proof}
		Similar to Lemma~\ref{lem:initarbitfault}, we need to ensure that all the robots move. So, the initial intermediate configuration $\mathcal{C}_0$, should not have points overlapping with the target pattern. According to the algorithms, we ensure that the target pattern points and the current configuration points do not have common points. We also ensure that the intersection points of the current configuration pattern and the target pattern are not present in the set of destination points. Even if the robots which are faulty already part of the target pattern, they do not lie on the uniform points. Hence we can determine which robots are faulty in the resulting~configuration.
	\end{proof}
	
	%\subsection{Determining the Pattern through faulty robots}
	%In this section, we show that a unique pattern exists that passes through all the faulty robots.
	\begin{lemma}\label{lem:unique}
		The target pattern passing through the faulty robots in $\mathcal{C}_1$ can be uniquely determined.
	\end{lemma}

	\begin{proof}
		The pattern is determined uniquely for a given value of $f$. For $f=2$ and $3$, the line and circle passing through the points are unique. For $f=4$, there can be two conjugate parabolas passing through four points. In this case, the parabola with the larger latus rectum is chosen as the target pattern. 
		For $f=5$, the target pattern is uniquely determined by the five points to be a parabola, hyperbola or an ellipse. Since we assume the faulty robots to form a convex polygon, they can only occupy positions on one side of the hyperbola.
		From Lemma~\ref{lem:nointersection}, we know that the destination points do not have a common point with the points in the previous configuration. From the quasi-uniform configuration, we can determine the robots which are not present at a uniform point. Hence we can determine the faulty robots in $\mathcal{C}_1$ and the corresponding pattern.
	\end{proof}

	\subsection{Termination}
	We can now show that the algorithm terminates, and we can determine the faulty robots in the terminal configuration.
	%Observe that 
	Combining Lemma~\ref{lem:initarbitfault},~\ref{lem:initinterfault} and~\ref{lem:unique} yields:
		\begin{theorem}\label{thm:twostep}
		Starting from any initial asymmetric configuration, the algorithm terminates in at most two rounds.
	\end{theorem}
	Since the algorithm does not do anything for a terminal configuration, we cannot determine the faulty robots if the initial configuration is a terminal configuration. 
	Moreover, in a terminal configuration, our algorithm designs result in the following distribution of robots on the plane starting from a configuration other than the terminal configuration. 
	\begin{corollary}\label{cor:uniformterminal}
		Starting from a configuration other than the terminal configuration, the non-faulty robots are at uniform pattern points in the terminal configuration.
	\end{corollary}
	\begin{proof}
		The destinations are always at uniform points spread over the target pattern. So whenever a non-faulty robot moves, it ends up at a uniform pattern point. Note that the resulting configuration may not be uniform due to the faulty robots. The non-faulty robots occupy the uniform points in a quasi-uniform configuration. 
	\end{proof}
	From Lemma~\ref{lem:nointersection} and Corollary~\ref{cor:uniformterminal}, we have the following Corollary.
	\begin{corollary}
		The faulty robots can be determined from a terminal configuration unless it is the initial configuration.
	\end{corollary}

	\section{Discussion}\label{sec:discussion}
	
	In the following, we show how to relax our model and extend the previous results in several directions. In particular, we extend the behavior of the algorithm in the absence of the assumption considered in \S\ref{sec:prelim}. We show that with small modifications to the algorithm, we can subvert some assumptions.

			\smallskip
		%\subsection{Knowledge of number of faults}\label{sec:atmostf}
			\noindent\textbf{Knowing the number of faults.}
			We extend the definition of Type I configuration to include configurations where at most $f$ robots are not in the pattern in the current configuration. As we need exactly $f$ robots to determine the target pattern, if $f' < f$ are not in the pattern, we choose a target pattern passing through those $f'$ and the first $f-f'$ robots in the ordering to set the~pattern. 
		%\subsection{Robots become faulty at any round}
			%In this case, we consider the same strategy of forming patterns as the \S\ref{sec:atmostf}. Since the algorithm always tries to form the pattern in one round, if a robot does not become faulty in that round, then the target pattern would be achieved. 
			%In the worst case, one robot would fail in one round, and we have to reform the target pattern in each round. But this can continue up to $f$ rounds. We can achieve the pattern formation in $f+1$ rounds. Note that if there are no failures in a particular round, then we achieve the pattern formation in the same round. If the robots recover from the crash-fault, they can participate in the algorithm and form the pattern, and this leads to self-stabilization. As long as there are no new faults in two consecutive rounds, the algorithm terminates.
			
			\smallskip
		%\subsection{Initial configuration with reflective symmetry}
			\noindent\textbf{Initial configuration with reflective symmetry.}
			For a configuration with a single line of symmetry, we can always follow the strategies described for Type I symmetric configurations in the algorithms from \S\ref{sec:algos}. The robots on the line of symmetry have two destinations on either side of the line of symmetry. According to their local orientation, they choose one of the destinations.

			\smallskip
		%\subsection{Lower order patterns for higher number of faults}
			\noindent\textbf{Lower order patterns for higher number of faults.}
			We add special cases if the robots are collinear (resp.\ co-circular)	for $f \in \{3, 4, 5\}$ (resp.\ $f\in\{4,5\}$). In this case, the robots form a line (resp.\ circle). If the initial configuration is this situation, then it is impossible to determine the faulty robots. Hence the configuration in the next step becomes an arbitrary configuration. We thus need three steps to achieve pattern formation instead of two.

	\section{Conclusion}\label{sec:conclusion}
	
	This paper initiated the study of distributed algorithms for pattern formation with faulty robots.
	In particular, we presented an algorithmic framework that allows solving many basic formation problems in at most two rounds, which is optimal given the lower bound also presented in this paper.
	We regard our work as a first step and believe it opens several interesting avenues for future research. In particular, it will be interesting to study pattern formation problems for 
	more advanced robots under failures, as well as randomized algorithms.
	It will also be interesting to generalize our failure model, e.g., to support transient crash faults and byzantine faults.

	% \section*{Acknowledgement}
	% The research work and the visit to the University of Vienna by the Indian authors are sponsored by Award No. ODF/2018/001055 under the Overseas Visiting Doctoral Fellowship scheme.
	% The authors are grateful to the Science and Engineering Research Board, Department of Science and Technology, Government of India, for their support. 

% \clearpage
	
	\bibliographystyle{splncs04}
	\bibliography{bib}

	\section*{Appendix}

	\subsection{Deferred Figures}
	
	In 
	%Fig.~\ref{fig:asymTypeIf5Appendix} 
	the following 
	we show the transitions for $f=5$. In all the figures, the target pattern is denoted by dashed curves and the pattern of the current configuration is denoted by solid curves. The current robot positions are represented by filled circles and target positions are denoted by empty circles. Fig.~\ref{fig:asymTypeIf5AppendixParabola} shows the transition from parabola to ellipse, parabola and hyperbola for Asymmetric Type I configurations.
	\begin{figure}[!h!t]
		\centering
		\includegraphics[width=0.3\linewidth]{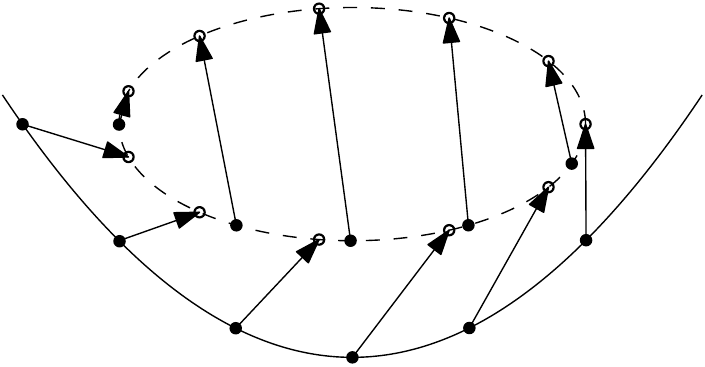}
		\includegraphics[width=0.3\linewidth]{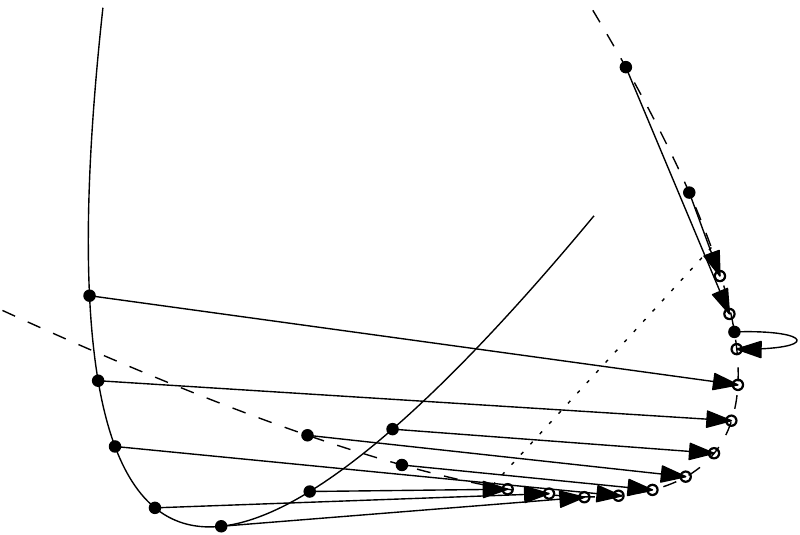}
		\includegraphics[width=0.3\linewidth]{figParabolaToHyperbola.pdf}
		\caption{Asymmetric Type I for $f=5$ starting from parabola}
		\label{fig:asymTypeIf5AppendixParabola}
	\end{figure}
	% Fig.~\ref{fig:asymTypeIf5AppendixHyperbola} shows the transition from hyperbola to ellipse, parabola and hyperbola for Asymmetric Type I configuration.
	% \begin{figure}[!ht]
	% 	\centering
	% 	\includegraphics[width=0.25\linewidth]{figParabolaToParabola5.pdf}
	% 	\includegraphics[width=0.25\linewidth]{figParabolaToEllipse.pdf}
	% 	\includegraphics[width=0.25\linewidth]{figParabolaToHyperbola.pdf}
	% 	\caption{Asymmetric Type I for $f=5$ starting from hyperbola}
	% 	\label{fig:asymTypeIf5AppendixHyperbola}
	% \end{figure}

	Fig.~\ref{fig:asymTypeIf5AppendixEllipse} shows the transition from ellipse to ellipse, parabola and hyperbola for Asymmetric Type I configurations.
	\begin{figure}[!h!t]
		\centering
		\includegraphics[width=0.3\linewidth]{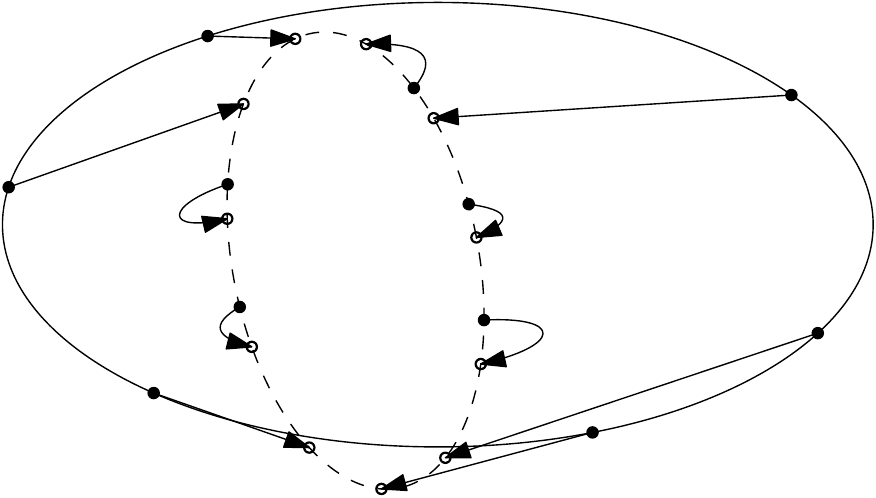}
		\includegraphics[width=0.3\linewidth]{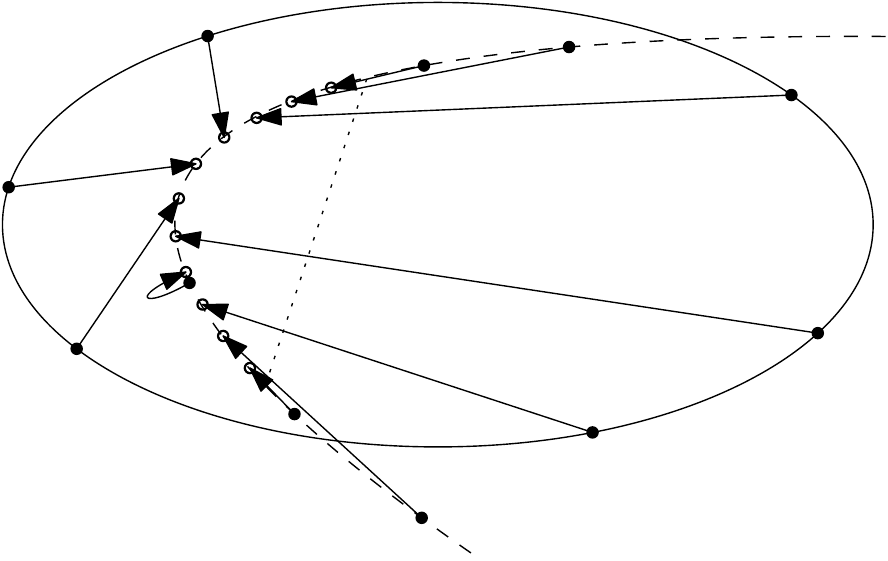}
		\includegraphics[width=0.3\linewidth]{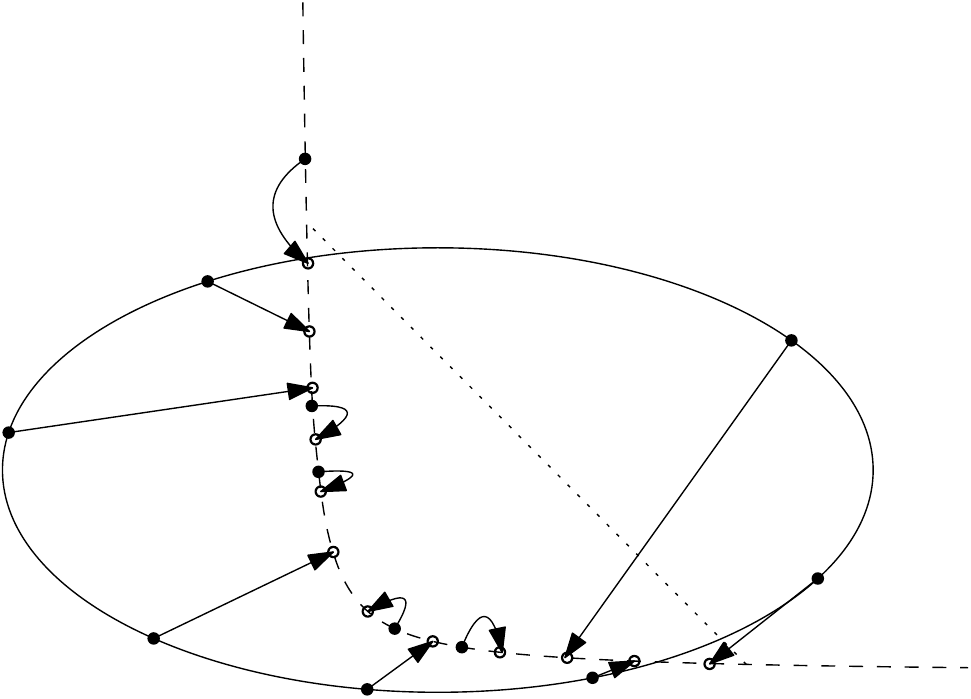}
		\caption{Asymmetric Type I for $f=5$ starting from ellipse}
		\label{fig:asymTypeIf5AppendixEllipse}
	\end{figure}

	\begin{figure}[!h!t]
		\centering
		\includegraphics[width=0.3\linewidth]{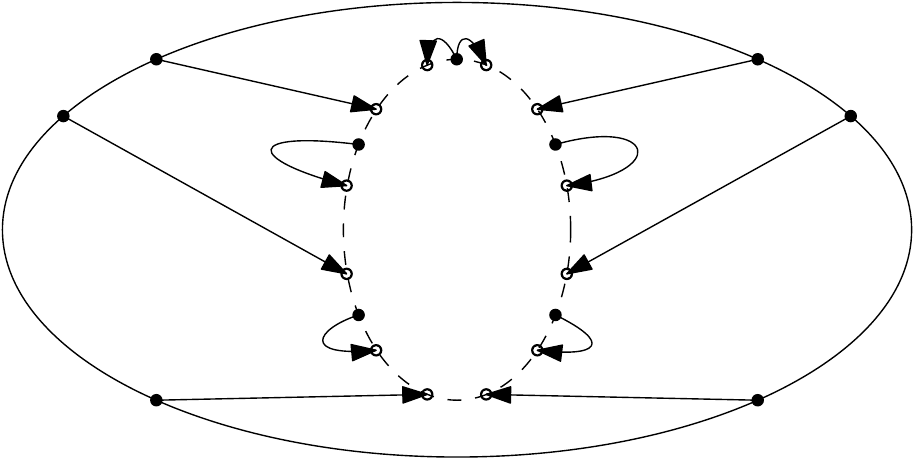}
		\includegraphics[width=0.3\linewidth]{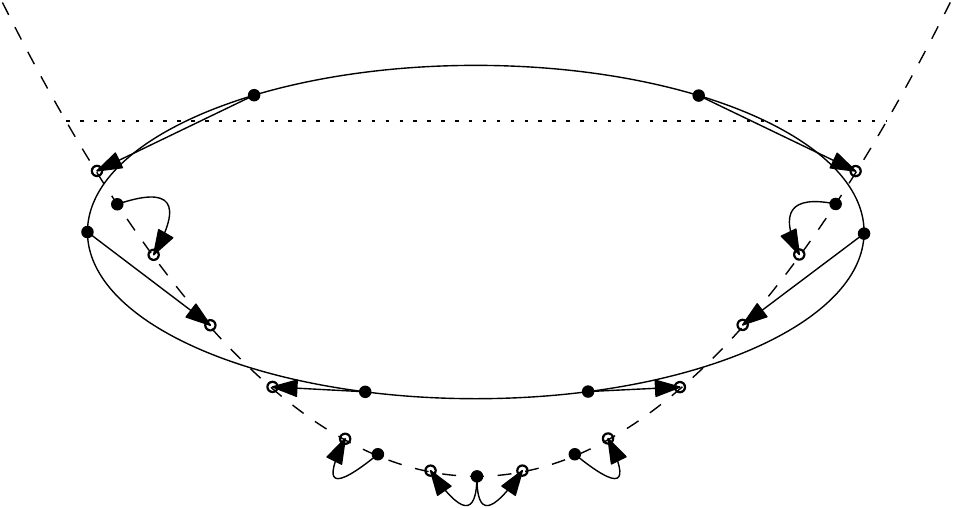}
		\caption{Reflective Symmetric Type I for $f=5$ starting from ellipse}
		\label{fig:SymTypeIf5AppendixEllipse}
	\end{figure}
	% 
	% Fig.~\ref{fig:SymTypeIf5Appendix} shows the transition from Hyperbola to Ellipse, Parabola and Hyperbola for Symmetric Type I configuration.
% 
	\begin{figure}[!h!t]
		\centering
		\includegraphics[width=0.3\linewidth]{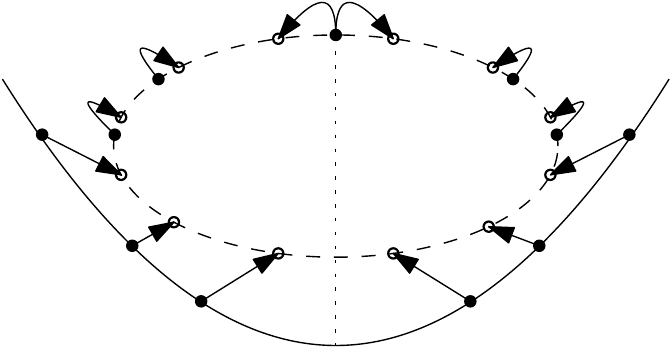}
		\includegraphics[width=0.3\linewidth]{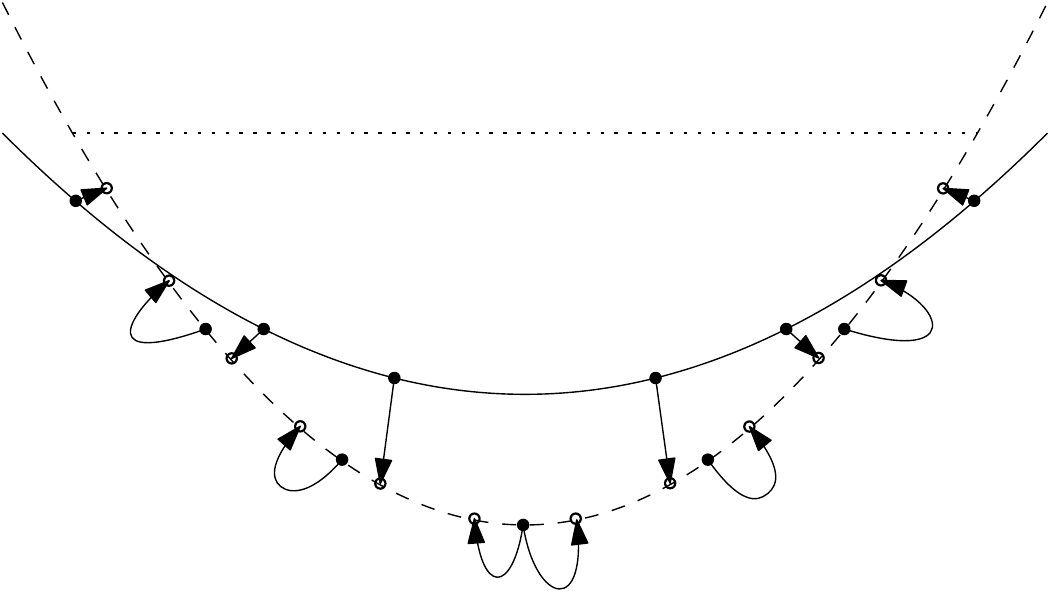}
		\caption{Reflective Symmetric Type I for $f=5$ starting from parabola}
		\label{fig:SymTypeIf5AppendixParabola}
	\end{figure}
	Fig.~\ref{fig:SymTypeIf5AppendixEllipse} shows the transition from ellipse to ellipse and from ellipse to parabola for Symmetric Type I configurations.
	Fig.~\ref{fig:SymTypeIf5AppendixParabola} shows the transition from parabola to ellipse, parabola to parabola for Symmetric Type I configurations.
	The strategies for hyperbola are similar to the strategies for parabola, since the faulty robots in a convex shape can only occupy one side of the hyperbola. Thus the target pattern becomes a curve between the intersection points of latus rectum with the hyperbola similar to the parabola. We omit some figures corresponding to hyperbola.
\end{document}